\documentclass[reqno]{amsart}
\usepackage{amsfonts}

\usepackage{tikz,amssymb,verbatim}

\renewcommand{\Im}{\operatorname{Im}}
\renewcommand{\Re}{\operatorname{Re}}

\newtheorem{theorem}{Theorem}[section]

\newtheorem{definition}[theorem]{Definition}

\newtheorem{remark}[theorem]{Remark}
\newtheorem{lemma}[theorem]{Lemma}
\newtheorem{example}[theorem]{Example}
\numberwithin{equation}{section}

\input{tcilatex}

\begin{document}
\title[KdV equation]{Singular Miura type initial profiles for the KdV
equation}
\author{Sergei Grudsky and Alexei Rybkin}
\address{Sergei M. Grudsky, Departamento de Matematicas, CINVESTAV del
I.P.N. Aportado Postal 14-740, 07000 Mexico, D.F., Mexico}
\email{grudsky@math.cinvestav.mx}
\address{Alexei Rybkin, Department of Mathematics and Statistics, University
of Alaska Fairbanks, PO Box 756660, Fairbanks, AK 99775}
\email{ arybkin@alaska.edu}
\date{August, 2011}
\thanks{Based on research supported in part by the NSF under grant DMS
1009673.}

\begin{abstract}
We show that the KdV flow evolves any real singular initial profile $q$ of
the form $q=r^{\prime }+r^{2}$, where $r\in L_{loc}^{2}$, $r|_{\mathbb{R}%
_{+}}=0$ into a meromorphic function with no real poles.
\end{abstract}

\maketitle

%\author{Alexei Rybkin}

%\tableofcontents

\section{Introduction}

This note is closely related to the recent paper \cite{KapPerryTopalov2005}
by Kappeler et al and \cite{Ryprep} by one of the authors.

More specifically, we are concerned with well-posedness (WP) of the initial
value problem (IVP) for the Korteweg-De Vries (KdV) equation ($x\in\mathbb{R}%
,t\ge0$) 
\begin{equation}  \label{eq1.1}
\begin{cases}
\partial_tu-6u\partial_xu+\partial_x^3u=0 \\ 
u(x,0)=q(x)%
\end{cases}%
\end{equation}
with certain low regularity non-decaying initial profiles $q$.

The problem of WP of \eqref{eq1.1} was raised back in the late 60's at about
the same time as the inverse scattering formalism for \eqref{eq1.1} was
discovered and has drawn an enormous attention. We are not in a position to
go over the extensive literature on the subject and refer to the book \cite%
{Tao06} by Tao where further literature is given.

The problem, of course, gets more difficult once we impose less regularity
on the initial data in \eqref{eq1.1}. Delta function type $q$'s in %
\eqref{eq1.1} were rigorously treated by Kappeler in \cite{Kap1986}. In the
present century, a large amount of effort has been put into WP in the
Sobolev space $H^{-s}$ with negative index\footnote{$H^{-s}$ is the space of
distributions subject to $(1+\left\vert x\right\vert )^{-s}\widehat{f}(x)\in
L^{2}$.}. The sharpest result is $s=3/4$ and due to Guo \cite{Guo09} which,
in turn, sharpens the result by Colliander et al \cite{ColKeStaTao03}. The
space $H^{-3/4}$ includes such singular functions as $\delta $, $1/x,$ etc.
However, the harmonic analytical methods employed in above papers break down
on $s=-1$. On the other hand, the Schr\"{o}dinger operator 
\begin{equation*}
L_{q}=-\partial _{x}^{2}+q(x)
\end{equation*}%
in the Lax pair associated with \eqref{eq1.1} is well-defined for $q\in
H^{-1}$ (see e.g. \cite{SS1999}) suggesting that the global WP could be
pushed to $H^{-1}$. It is exactly how Kappeler-Topalov \cite{KapTop06} were
able to extend WP to $H^{-1}(\mathbb{T})$ for periodic $q$'s. It is natural
to conjecture that the global WP for \eqref{eq1.1} also holds and could be
achieved by a suitable extension of the inverse scattering transform (IST)
method for $L_{q}$ with $q\in H^{-1}$. An important step in this direction
was done by Kappeler et al \cite{KapPerryTopalov2005} where it was shown
that \eqref{eq1.1} is globally well-posed in a certain sense if $q=r^{\prime
}+r^{2}$ with some $r\in L^{2}$. The transform 
\begin{equation*}
B(r)=r^{\prime }+r^{2}\quad ,\quad r\in L_{loc}^{2}
\end{equation*}%
is called Miura. Of course, $B(L^{2})$ doesn't exhaust $H^{-1}$ as $H^{-1}$
consists of all functions $f=r^{\prime }+p$ with some $r,p\in L^{2}$. It is
easy to see that $L_{q}$ is a positive operator for any $q=r^{\prime }+r^{2}$
(we call such $q$ a Miura potential) meaning that \eqref{eq1.1} admits
so-called `dispersive' solutions, i.e. solutions which disperse in time and
do not have a soliton component. However, since $H^{-1}\supset H^{-3/4}$,
singularity of such solutions is pushed all the way to\footnote{%
As indicated in \cite{SS1999}, $L_{q}$ with $q\in H^{-s}$ for $s>1$ is
ill-defined.} $s=-1$.

We note that all functions in $H^{-s}(\mathbb{R})$ exhibit certain decay at $%
\pm\infty$. On the other hand, there has been a significant interest in
non-decaying solutions to \eqref{eq1.1} (other than periodic). The case of
the so-called steplike initial profiles (i.e. when $q(x)\to0$ sufficiently
fast as $x\to+\infty$ ($-\infty$) and $q(x)$ doesn't decay at $-\infty$ ($%
+\infty$)). is of physical interest and has attracted much attention since
the early 70s. We refer to the recent paper \cite{EGT09} by
Egorova-Grunert-Teschl for a comprehensive account of the (rigorous)
literature on steplike initial profiles with specified behavior at infinity
(e.g. $q$'s tending to a constant, periodic function, etc.). In the recent
preprint of one of the authors \cite{Ryprep} (see also \cite{Ryb10}), the
case of $q$'s rapidly decaying at $+\infty$ and sufficiently arbitrary at $%
-\infty$ is studied in great detail. Initial steplike profiles in these
papers are at least locally integrable (i.e. regular).

The current note is concerned with treating Miura steplike initial data $q$
in \eqref{eq1.1}. Namely, we consider $q=r^{\prime }+r^{2}$ for $r\in
L_{loc}^{2}$ identically (for simplicity) vanishing on $(0,\infty )$. Even
though $q$ has very low regularity and is essentially arbitrary on $(-\infty
,0)$ the fact that $q$ is zero on $(0,\infty )$ leads to an extremely strong
smoothing effect. Dispersion instantaneously turns such initial profiles $%
q(x)$ into a function $u(x,t)$ meromorphic in $x$ on the whole complex plane
for any $t>0$. The WP of the problem \eqref{eq1.1} can therefore be
understood in a classical sense and moreover it comes with an explicit
formula 
\begin{equation*}
u(x,t)=-2\partial _{x}^{2}\log \det \left( 1+\mathbb{H}_{x,t}\right) ,
\end{equation*}%
where $\mathbb{H}_{x,t}$ is the Hankel operator with symbol 
\begin{equation*}
\varphi _{x,t}(\lambda )=\frac{i\lambda -m(\lambda ^{2})}{i\lambda
+m(\lambda ^{2})}e^{2i\lambda (4\lambda ^{2}t+x)},\quad \lambda \in \mathbb{R%
}\;,\;x\in \mathbb{R}\;,\;t>0,
\end{equation*}%
where $m$ is the Titchmarsh-Weyl $m$-function associated with $L_{q}$ on $%
(-\infty ,0)$ with a Dirichlet boundary condition at $0$.

The WP of our problem, among others, means that $u(x,t)$ has no real poles
for any $t>0$. I.e. no positon solution may occur in our situation.

Our approach is based on a suitable adaptation of the IST and analysis of
Hankel operators with oscillatory symbols. To keep our note as short as
possible, we will omit some technical issues and come back to them elsewhere
in a more suitable setting.

The paper is organized as follows. In Section 2 we review Hankel operators
and prove a new result related to a Hankel operator with a cubic oscillatory
symbol. In Section 3 we discuss the Titchmarsh-Weyl $m$-function and
reflection coefficient in the context of singular points. In the last
Section 4 we state and prove our main result.

%\section{Notation and preliminaries}

\section{Hankel operators}

Hankel operators naturally appear in linear algebra, operator theory,
complex analysis, mathematical physics, and many other areas. In our note
they play a crucial role. However, their formal definitions vary. In the
context of integral operators, a Hankel operator is usually defined as an
integral operator on $L^2_+:=L^2(\mathbb{R}_+)$ whose kernel depends on the
sum of the arguments. I.e. 
\begin{equation}  \label{eqHO1.1}
\left(\mathbb{H}f\right)(x)= \int_0^\infty H(x+y)f(y)dy \quad,\quad x\ge0
\;,\; f\in L^2_+\;,
\end{equation}
with some function $H$.

In many situations, including ours, $H$ is not a function but rather a
distribution. It is convenient then to accept a regularized version of %
\eqref{eqHO1.1}.

Let\footnote{%
For brevity we set $\int:=\int_{-\infty}^\infty$} 
\begin{equation*}
\left(\mathcal{F}f\right)(\lambda)= \frac{1}{\sqrt{2\pi}}\int e^{i\lambda
x}f(x)dx
\end{equation*}
be the Fourier transform and $\chi$ the Heaviside function of $\mathbb{R}%
_+:=(0,\infty)$.

\begin{definition}
\label{defHO1} Given $\varphi\in L^\infty$, we call the operator $\mathbb{H}%
_\varphi$ on $L^2_+$ defined for any $f\in L^2_+$ by 
\begin{equation}  \label{eqHO2.1}
\mathbb{H}_\varphi f= \chi\mathcal{F}\varphi\mathcal{F}f
\end{equation}
the Hankel operator on $L^2_+$ with symbol $\varphi$.
\end{definition}

It follows from a straightforward computation that \eqref{eqHO1.1} and %
\eqref{eqHO2.1} agree if $\varphi\in L^2\cap L^\infty$ and $H=\mathcal{F}%
\varphi$. However if $\varphi$ is merely $L^\infty$ then $\mathcal{F}\varphi$
is not a function but a (tempered) distribution. The operator $\mathbb{H}$
given by \eqref{eqHO1.1} is no longer well-defined. But the one given by %
\eqref{eqHO2.1} is.

The Hankel operator $\mathbb{H}_\varphi$ is clearly bounded from %
\eqref{eqHO2.1}. One immediately has 
\begin{equation}  \label{eqHO2.2}
\left\Vert \mathbb{H}_\varphi \right\Vert \le \left\Vert \varphi
\right\Vert_\infty.
\end{equation}

Membership of $\mathbb{H}_\varphi$ in narrower Schatten-Von Neumann ideals
is, however, a much more subtle issue which was completely resolved by
Peller in about 1980 (see e.g. \cite{Peller2003}).

We will be particularly concerned with the invertibility of $1+\mathbb{H}%
_\varphi$. The first fact is trivial.

\begin{lemma}
\label{lem1}Let $\varphi$ be such that $\left\vert \varphi(\lambda)
\right\vert\le1$ a.e. $\lambda\in\mathbb{R}$ and $\left\vert
\varphi(\lambda) \right\vert<1$ a.e. on a set $S$ of positive Lebesgue
measure. Then $-1$ is not an eignevalue of $\mathbb{H}_\varphi$.
\end{lemma}

\begin{proof}
Assume $-1$ is an eigenvalue of $\mathbb{H}_\varphi$ and $f\ne0$ is the
corresponding normalized eigenvector (i.e. $\left\Vert f
\right\Vert_{L^2_+}=1$).

It follows from 
\begin{equation*}
f+\mathbb{H}_\varphi f=0
\end{equation*}
that 
\begin{equation*}
1+\int \varphi(\lambda)\widehat{f}(\lambda)\widehat{\overline{f}}%
(\lambda)d\lambda=0
\end{equation*}
and hence 
\begin{equation}  \label{eqHO3.1}
1+\Re \int\varphi(\lambda)\widehat{f}(\lambda)\widehat{\overline{f}}%
(\lambda)d\lambda=0.
\end{equation}

But 
\begin{align}
\Re \int \varphi (\lambda )\widehat{f}(\lambda )d\lambda & \leq \int
\left\vert \varphi (\lambda )\widehat{f}(\lambda )\right\vert \cdot
\left\vert \widehat{\overline{f}}(\lambda )\right\vert d\lambda   \notag \\
& \leq \left\Vert \varphi \widehat{f}\right\Vert _{L^{2}}\cdot \left\Vert 
\widehat{\overline{f}}\right\Vert _{L^{2}}=\left\Vert \varphi \widehat{f}%
\right\Vert _{L^{2}}  \notag \\
& \leq \int_{S}\left\vert \varphi (\lambda )\right\vert ^{2}\cdot \left\vert 
\widehat{f}(\lambda )\right\vert ^{2}d\lambda   \notag \\
& <\int_{S}\left\vert \widehat{f}(\lambda )\right\vert ^{2}\leq 1.
\label{eqHO3.2}
\end{align}

Comparing \eqref{eqHO3.1} and \eqref{eqHO3.2} leads to a contradiction.
\end{proof}

The proof of Lemma \ref{lem1} is no longer valid if $\left\vert
\varphi(\lambda) \right\vert=1$ for a.e. real $\lambda$.

However in our setting symbols $\varphi$ have a very specific structure 
\begin{equation}  \label{eqHO4.3}
\varphi(\lambda)= e^{i\lambda(\lambda^2+a)}I(\lambda)
\end{equation}
where $a$ is a real number and $I$ is an inner function of the upper half
plane (i.e.\footnote{$H^p_\pm$ ($0<p\le\infty$) are standard Hardy spaces of
the upper (lower) half planes $\mathbb{C}^\pm$.} $I\in H_+^\infty$ and $%
\left\vert I(\lambda) \right\vert=1$ a.e. $\lambda\in\mathbb{R}$).

\begin{lemma}
\label{lem2} Let $\varphi$ be given by \eqref{eqHO4.3}. Then

\begin{enumerate}
\item \label{it1} $\mathbb{H}_\varphi$ is a compact operator,

\item \label{it2} $1+\mathbb{H}_\varphi$ is invertible.
\end{enumerate}
\end{lemma}

\begin{proof}
Our argument is based upon the factorization (see \cite{BGS2001}, \cite%
{DyGrud02} Section 5.10, \cite{Gru01} ) 
\begin{equation}
e^{i\lambda (\lambda ^{2}+a)}=B(\lambda )U(\lambda )\quad ,\quad \lambda \in 
\mathbb{R},  \label{eqHO5.1}
\end{equation}%
where $B(\lambda )$ is a Blaschke product with infinitely many zeros
accumulating at infinity and $U$ is a unimodular function from $C(\overline{%
\mathbb{R}})$, the class of continuous on $\mathbb{R}$ functions $f$ subject
to 
\begin{equation*}
\lim_{\lambda \rightarrow -\infty }f(\lambda )=\lim_{\lambda \rightarrow
\infty }f(\lambda )\neq \pm \infty .
\end{equation*}

Since a product of an inner function and a $C(\overline{\mathbb{R}})$%
-function is in the algebra $H_{+}^{\infty }+C(\overline{\mathbb{R}})$, by
the Hartman theorem \cite{Nik2002} $\mathbb{H}_{\varphi }$ is compact and %
\eqref{it1} is proven.

Consider the Hankel operator \eqref{eqHO2.1} in the Fourier representation.
Denoting $P_{\pm }$ the Riesz projection in $L^{2}$ onto $H_{\pm }^{2}$, we
have 
\begin{align*}
\mathcal{F}\mathbb{H}_{\varphi }\mathcal{F}^{-1}& =\mathcal{F}\chi \mathcal{F%
}\varphi \mathcal{F}\mathcal{F}^{-1} \\
& =P_{+}\mathcal{F}\mathcal{F}\varphi =P_{+}\mathcal{F}^{2}\varphi  \\
& =P_{+}J\varphi =JP_{-}\varphi 
\end{align*}%
where $Jf(x)=f(-x)$. Thus, the operator 
\begin{equation}
JP_{-}\varphi :H_{+}^{2}\rightarrow H_{+}^{2}  \label{eqHO7.1}
\end{equation}%
is unitarily equivalent to $\mathbb{H}_{\varphi }$. Let $\mathbb{T}_{\varphi
}$ be the Toeplitz operator on $H_{+}^{2}$. I.e. 
\begin{equation*}
\mathbb{T}_{\varphi }f=P_{+}\varphi f,\quad f\in H_{+}^{2}.
\end{equation*}

Note (\cite{BotSil06} Ch. 2) that (\ref{eqHO5.1}) implies left-invertibility
of the operator $\mathbb{T}_{\varphi }$ and, by the Devinatz-Widom theorem (%
\cite{BotSil06} p. 59), there exists a function $f\in H_{+}^{\infty }$, such
that 
\begin{equation*}
\left\Vert \varphi -f\right\Vert _{L^{\infty }}<1.
\end{equation*}%
Thus, it immediately follows from the representation (\ref{eqHO7.1}) that 
\begin{equation*}
\mathbb{H}_{\varphi }=\mathbb{H}_{\varphi -f}
\end{equation*}%
and hence (\ref{eqHO2.2}) and (\ref{eqHO5.1}) 
\begin{equation*}
\left\Vert H_{\varphi }\right\Vert =\left\Vert \mathbb{H}_{\varphi
-f}\right\Vert \leq \left\Vert \varphi -f\right\Vert _{L^{\infty }}<1.
\end{equation*}

This proves \eqref{it2} and the lemma is proven.
\end{proof}

%\section{Titchmarsh-Weyl $m$-function}

\section{The Titchmarsh-Weyl $m$-function and the reflection coefficient}

Denote $H^{-1}_{loc}:=H^{-1}_{loc}(\mathbb{R})$ the local $H^{-1}$ space
(i.e. the set of all functions $\widetilde{\chi}_Sf$, where $f\in H^{-1}$
and $\widetilde{\chi}_S$ is a smoothened characteristic function of a
compact set $\mathbb{R}$). It is well-known that any $q \in H^{-1}_{loc}(%
\mathbb{R})$ can be represented as $q=Q^{\prime }$ with some $Q\in L^2_{loc}$
and we rewrite 
\begin{equation*}
-y^{\prime \prime }+qy=z y
\end{equation*}
as 
\begin{equation}  \label{eq3.0}
-(y^{\prime }-Qy)^{\prime }-Qy^{\prime }=z y
\end{equation}
(the regularized Schr\"{o}dinger equation).

Following the approach of \cite{SS1999} we introduce 
\begin{equation}  \label{eq3.1}
\begin{cases}
L_q := -\partial_x(\partial_x-Q)-Q\partial_x \\ 
\partial_x Q=q%
\end{cases}%
\end{equation}
the Schr\"{o}dinger operator with a (singular) potential $q\in H^{-1}_{loc}(%
\mathbb{R})$.

As proven in \cite{SS1999}, the operator \eqref{eq3.1} is well-defined. One
can also extend the classical Titchmarsh-Weyl theory to $L_q$. In
particular, the Weyl limit point/circle classification can be easily
extended to singular $q$'s. We plan to provide the details elsewhere and
only mention here that regular derivatives $\partial_x$ in classical
Titchmarsh-Weyl theory should, where appropriate, be replaced by ``quasi"
derivative $\partial_x-Q$. Note that this doesn't change the Wronskian as 
\begin{equation*}
\det 
\begin{pmatrix}
y_1 & y_2 \\ 
y_1^{\prime }-Qy_1 & y_2^{\prime }-Qy_2%
\end{pmatrix}%
= \det 
\begin{pmatrix}
y_1 & y_2 \\ 
y_1^{\prime } & y_2^{\prime }%
\end{pmatrix}%
\end{equation*}
if $y_1,y_2$ are a.c.\footnote{%
a.c. abbreviates absolutely continuous.}

Let's now define the (Dirichlet) Titchmarsh-Weyl $m$-function corresponding
to $\mathbb{R}_-$. Assuming that $q=Q^{\prime }$, with some $Q\in L^2_{loc}(%
\mathbb{R})$ and $L_q$ is limit point case at $-\infty$ and $Q|_{\mathbb{R}%
_+}=0$.

Denoting $\psi(x,z)$ the Weyl solution (i.e. $\psi\in L^2(\mathbb{R}_-)$ for
any $z\in\mathbb{C}^+$ of \eqref{eq3.0} with $Q\in L^2_{loc}$ and $Q|_{%
\mathbb{R}_+}=0$) we define the Titchmarsh-Weyl $m$-function as 
\begin{equation}  \label{eq4.1}
m(z)= - \frac{\partial_x\psi(+0,z)}{\psi(+0,z)}.
\end{equation}
Note that %although 
$\partial_x\psi(x,z)$ is not a.c. for $x\ge0$ (whereas $\partial_x\psi-Q\psi$
is), but $\partial_x\psi(x,z)=\partial_x\psi(x,z)-Q(x)\psi(x,z)$ for $x>0$
and $\partial_x\psi(+0,z)$ are well-defined. As its regular counterpart, the
Titchmarsh-Weyl $m$-function has the following properties:

Properties of $m$. %(place before RC)

\begin{enumerate}
\item \label{pr1} $m$ is analytic and Herglotz. I.e. $m:\mathbb{C}^+\to%
\mathbb{C}^+$.

\item \label{pr2} Let $Q_n$ be a sequence of smooth $L^2_{loc}$ functions
such that $\left\Vert Q-Q_n \right\Vert_{L^2_{loc}}\to0$, $n\to\infty$, and $%
q=Q^{\prime }$ is limit point case at $-\infty$. Then ${m}_n\to m$ uniformly
on compact subsets of $\mathbb{C}^+$. %\rightrightarrows m$
\end{enumerate}

Define now the reflection coefficient $R$ from the right incident of a
singular potential $q\in H^{-1}_{loc}(\mathbb{R})$ such that $q|_{\mathbb{R}%
_+}=0$.

Pick up a point $x_0>0$ and consider a solution to $L_qy=\lambda^2y$ which
is proportional to the Weyl solution on $(-\infty,x_0)$ and is equal to $%
e^{-i\lambda x}+re^{i\lambda x}$ on $(x_0,\infty)$. From the continuity of
this solution and its derivative at $x_0$ one has 
\begin{equation*}
r(\lambda,x_0)= e^{-2i\lambda x_0} \frac{i\lambda-\frac{\psi^{\prime
}(x_0,\lambda^2)}{\psi(x_0,\lambda^2)}}{i\lambda+\frac{\psi^{\prime
}(x_0,\lambda^2)}{\psi(x_0,\lambda^2)}}.
\end{equation*}

We define the right reflection coefficient by 
\begin{equation}  \label{eq3.3}
R(\lambda)=\lim_{x_0\to0^+}r(\lambda,x_0)= \frac{i\lambda-m(\lambda^2)}{%
i\lambda+m(\lambda^2)}.
\end{equation}

\begin{example}
\label{ex1} Let $q(x)=c \delta(x)$. The Weyl solution corresponding to $%
-\infty$ can be explicitly computed by ($C\ne0$) 
\begin{equation*}
\psi(x,\lambda^2)= C 
\begin{cases}
e^{-i\lambda x} \quad & ,\quad x<0 \\ 
\frac{1}{2i\lambda} \left( ce^{i\lambda x}+(2i\lambda-c)e^{-i\lambda
x}\right) \quad & , \quad x>0%
\end{cases}%
\end{equation*}
and hence by \eqref{eq4.1} and \eqref{eq3.3} 
\begin{align*}
m(\lambda^2) &= i\lambda-c, \\
R(\lambda) &= \frac{c}{2i\lambda-c}.
\end{align*}
\end{example}

\section{Main result}

In the last section we state and prove our main result. As customary, given
self-adjoint operator $A$ we write $A\ge0$ if $A$ is positive.

\begin{theorem}
\label{thm1} Let $q$ in \eqref{eq1.1} supported on $\mathbb{R}_{-}$ be in $%
H^{-1}$ and such that the Schr\"{o}dinger operator $L_{q}\geq 0$. Then there
is a (unique) classical solution to \eqref{eq1.1} given by 
\begin{equation}
u(x,t)=-2\partial _{x}^{2}\log \det \left( 1+\mathbb{H}_{x,t}\right) 
\label{eq2.1}
\end{equation}%
where $\mathbb{H}_{x,t}$ is the trace class Hankel operator on $L^{2}(%
\mathbb{R}_{+})$ with the symbol 
\begin{equation*}
\varphi _{x,t}(\lambda )=\frac{i\lambda -m(\lambda ^{2})}{i\lambda
+m(\lambda ^{2})}e^{2i\lambda x+8i\lambda ^{3}t}
\end{equation*}%
where $m$ is the (Dirichlet) Titchmarsh-Weyl $m$-function of $L_{q}$ on $%
L^{2}(\mathbb{R}_{-})$.

The solution $u(x,t)$ is meromorphic in $\mathbb{C}^+$ for any $t>0$ except
(double) poles none of which are real.
\end{theorem}

\begin{proof}
It is proven in \cite{KapPerryTopalov2005} that 
\begin{equation*}
L_{q}\geq 0\quad \Rightarrow \quad q\in B\left( L_{loc}^{2}\right) \subset
H_{loc}^{-1}
\end{equation*}%
where $B(r)=r^{\prime }+r^{2}$ is the Miura map.

Since compactly supported smooth functions are dense in $H^{-1}_{loc}$, we
can approximate our $q$ by a sequence $\tilde{q}=\tilde{r}^{\prime }+\tilde{r%
}^2$ where $\tilde{r}$'s are smooth and compactly supported.

For each $\tilde{q}$ there exists the (classical) right reflection
coefficient $\widetilde{R}$. The (classical) Marchenko operator $\widetilde{%
\mathbb{H}}_{x,t}$ has no discrete component (since $L_{\tilde{q}} \ge0$)
and hence it takes the form 
\begin{equation}  \label{eq7.1}
\left(\widetilde{\mathbb{H}}_{x,t}f\right)(\cdot)= \int_0^\infty \widetilde{H%
}_{x,t}(\cdot+y)f(y)dy
\end{equation}
where 
\begin{equation}  \label{eq7.2}
\widetilde{H}_{x,t}(\cdot)= \frac{1}{2\pi} \int e^{2i\lambda
x+8i\lambda^3t}e^{i\lambda(\cdot)}\widetilde{R}(\lambda)d\lambda.
\end{equation}

The reflection coefficient $\widetilde{R}$ can be computed by 
\begin{equation*}
\widetilde{R}(\lambda)= \frac{i\lambda-\widetilde{m}(\lambda^2)}{i\lambda+%
\widetilde{m}(\lambda^2)}
\end{equation*}
where $\tilde{m}$ is the Titchmarsh-Weyl $m$-function of $L^0_{\tilde{q}}$,
the Dirichlet $-\partial_x^2+\tilde{q}(x)$ on $\mathbb{R}_-$. Since the
function $\widetilde{R}(\lambda)$ is analytic in $\mathbb{C}^+$ and $%
\widetilde{R}(\lambda)=O(1/\lambda) \; , \; \lambda\to\pm\infty$, and $%
\left\vert \widetilde{R}(\lambda) \right\vert\le1\;,\; \lambda\in\mathbb{C}^+
$, one can obviously deform the contour of integration in \eqref{eq7.2} and %
\eqref{eq7.2} reads 
\begin{equation}  \label{eq8.1}
\widetilde{H}_{x,t}(\cdot)= \frac{1}{2\pi} \int_{\Im\lambda=h} e^{2i\lambda
x+8i\lambda^3t}e^{i\lambda(\cdot)}\widetilde{R}(\lambda)d\lambda.
\end{equation}
for any $h>0$. Since the integrand in \eqref{eq8.1} is clearly integrable
along the line $\Im\lambda=h$, the operator $\widetilde{\mathbb{H}}_{x,t}$
is trace class (see \cite{Ryprep}) and the function 
\begin{equation}  \label{eq8.2}
\tilde{u}(x,t)=-2\partial_x^2\log\det\left(1+\widetilde{\mathbb{H}}%
_{x,t}\right)
\end{equation}
is well-defined and solves \eqref{eq1.1} with initial data $\tilde{q}$.

We now pass to the limit in \eqref{eq8.2} as $\tilde{r}\to r$ in $L^2_{loc}$%
. By property \eqref{pr2} of the Titchmarsh-Weyl $m$-function, 
\begin{equation*}
\widetilde{R}(\lambda)= \frac{i\lambda-\widetilde{m}(\lambda^2)}{i\lambda+%
\widetilde{m}(\lambda^2)} \quad \longrightarrow \quad {R}(\lambda)= \frac{%
i\lambda-{m}(\lambda^2)}{i\lambda+{m}(\lambda^2)}
\end{equation*}
on each compact set in $\mathbb{C}^+$. The oscillatory factor $e^{2i\lambda
x+8i\lambda^3t}$ exhibits a superexponential decay on $\Im\lambda=h>0$. This
means that (see \cite{Ryprep} for) 
\begin{equation*}
\widetilde{\mathbb{H}}_{x,t} \; \longrightarrow \; \mathbb{H}_{x,t}
\end{equation*}
for any $x\in\mathbb{R}$, $t>0$ in trace class norm and hence 
\begin{equation*}
\det\left(1+\widetilde{\mathbb{H}}_{x,t}\right) \; \longrightarrow \;
\det\left(1+\mathbb{H}_{x,t}\right).
\end{equation*}

Note that $\widetilde{H}_{x,t}$ and 
\begin{equation*}
{H}_{x,t}(\cdot)= \frac{1}{2\pi} \int_{\Im\lambda=h} e^{2i\lambda
x+8i\lambda^3t}e^{i\lambda(\cdot)}{R}(\lambda)d\lambda
\end{equation*}
are clearly entire with respect to $x$, $\forall\;t>0$. It is quite easy to
see that $\widetilde{\mathbb{H}}_{x,t}, \mathbb{H}_{x,t}$ are
operator-valued functions entire with respect to $x$, $\forall\;t>0$. This
means that the functions 
\begin{equation*}
\tilde{u}(x,t)=-2\partial_x^2\log\det\left(1+\widetilde{\mathbb{H}}%
_{x,t}\right)
\end{equation*}
are meromorphic in $x$ on the whole complex plane for any $t>0$ and converge
to the meromorphic function 
\begin{equation*}
{u}(x,t)=-2\partial_x^2\log\det\left(1+{\mathbb{H}}_{x,t}\right)
\end{equation*}
as $\tilde{r}\to r$ in $L^2_{loc}$.

It remains to show that $\det(1+\mathbb{H}_{x,t})$ doesn't vanish on the
real line for any $t>0$. Since $\mathbb{H}_{x,t}$ is trace class, this
amounts to showing that $-1$ is not an eigenvalue of $\mathbb{H}_{x,t}$ for
all $x\in\mathbb{R}\;,\;t>0$. We have two cases: $L_q$ has some a.c.
spectrum, $L_q$ has no a.c. spectrum. The first case immediately follows
from Lemma \ref{lem1}.

The second case is a bit more involved. If the a.c. spectrum of $L_q$ is
empty then the Titchmarsh-Weyl $m$-function is real a.e. on the real line
and hence the reflection coefficient $\left\vert R(\lambda) \right\vert\le1$
in $\mathbb{C}^+$ and $\left\vert R(\lambda) \right\vert=1$ a.e. on $\mathbb{%
R}$. I.e. $R$ is an inner function of the upper half plane. Lemma \ref{lem2}
then applies.
\end{proof}

\begin{remark}
\label{rem1} Theorem \ref{thm1} implies very strong WP of the KdV equation
with eventually any steplike Miura initial data supported on $(-\infty,0)$.
Each such solution $u(x,t)$ is smooth and hence solves the KdV equation in
the classical sense. It also has a continuity property in the sense that if $%
\{q_n\}$ is a sequence of smooth $H^{-1}_{loc}$ functions convergent in $%
H^{-1}_{loc}$ to $q$ then the sequence of the corresponding solutions $%
\{u_n(x,t)\}$ converges in $H^{-1}_{loc}$ to $u(x,t)$. This, in turn,
implies uniqueness. The initial condition is satisfied in the sense that 
\begin{equation*}
\left\Vert u(\cdot,t)-q \right\Vert_{H^{-1}_{loc}}\;\to\; 0 \quad,\quad
t\to0.
\end{equation*}
\end{remark}

\begin{remark}
\label{rem2}It is unlikely that, under our conditions, $\mathbb{H}_{x,t}$ in %
\eqref{eq2.1} is trace class for any $x$ if $t=0$. We conjecture however
that if $Q$ is uniformly in $L_{loc}^{2}$, i.e. $\sup_{x\leq
0}\int_{x-1}^{x}\left\vert Q\right\vert ^{2}<\infty $, then $\mathbb{H}_{x,0}
$ is also trace class for any real $x$. 
\end{remark}

\begin{remark}
\label{rem3} We assumed $q|_{\mathbb{R}_{+}}=0$ for simplicity and it can be
replaced with a suitable decay condition but the consideration becomes much
more involved due to serious technical circumstances. We plan to return to
it elsewhere.
\end{remark}

\end{document}